\documentclass[preprint,12pt]{elsarticle}

\usepackage{amssymb,times,amsmath,xspace,dsfont,pgf,latexsym,natbib,cmll}
 \usepackage{mathtools}
\usepackage[]{graphicx}
\usepackage[all]{xy}
\usepackage{enumitem}
\usepackage{mathrsfs}
\usepackage{color}
\usepackage{tikz-cd}
\usetikzlibrary{cd}

\usepackage{lineno}

\newtheorem{definition}{Definition}[section]

\newtheorem{remark}{Remark}[section]
\newtheorem{lemma}{Lemma}[section]
\newtheorem{corollary}{Corollary}[section]
\newtheorem{proposition}{Proposition}[section]
\newtheorem{theorem}{Theorem}[section]
\newenvironment{proof}{\noindent\textbf{Proof. }}{\hfill \small $\Box$}
\journal{xxxxxxxxxx}

\begin{document}
	
	\begin{frontmatter}
		
		
		
		\title{Residuated implications derived from quasi-overlap functions on lattices}
		
		\author[label1]{Rui Paiva}
		\author[label2]{Benjam\'in Bedregal}
		\author[label2]{Regivan Santiago}
		\address[label1]{Instituto Federal de Educa\c{c}\~ao, Ci\^{e}ncia e Tecnologia do Cear\'a\\
			Maracana\'u, Brazil\\
			Email: rui.brasileiro@ifce.edu.br}

		\address[label2]{Departamento de Inform\'atica e Matem\'atica Aplicada  \\
			Universidade Federal do Rio Grande do Norte \\
			Natal, Brazil\\
			Email: \{bedregal,regivan\}@dimap.ufrn.br}
		
		
		\begin{abstract}
		\noindent
		In this paper, we introduce the concept of residuated implications derived from quasi-overlap functions on lattices and prove some related properties. In addition, we formalized the residuation principle for the case of quasi-overlap functions on lattices and their respective induced implications, as well as revealing that the class of quasi-overlap functions that fulfill the residuation principle is the same class of continuous functions according to topology of Scott. Also, Scott's continuity and the notion of densely ordered posets are used to generalize a classification theorem for residuated quasi-overlap functions. Finally, the concept of automorphisms are extended to the context of quasi-overlap functions over lattices, taking these lattices into account as topological spaces, with a view to obtaining quasi-overlap functions conjugated by the action of automorphisms.
		\end{abstract}
		
		
	\end{frontmatter}
	
	

\section{Introduction}

Overlap functions were introduced by Bustince et al. \cite{BUSTINCE} as a class of aggregate functions with two entries over the interval $ [0,1] $ to be applied to the image processing field. Basically, these functions transform pixel images with values at $[0,1]$. Many researchers have began to develop the theory of overlap functions to explore their
potentialities in different scenarios, such as problems involving classification or
decision making \cite{GOMEZ201657,DEMIGUEL2018, NOLASCO_recomen, QIAO20182, QIAO2018,QIAO20181,WANG2018} and  from the theoretical point of view \cite{BEDREGAL20171,DIMURO3,DIMURO2,PAIVA,ZHANG2020}. However, when you consider that pixels (or signs) may contain uncertainties, for example noise, this noise information can be captured on objects
that extend real numbers, for example intervals, fuzzy numbers or interval-valued fuzzy sets, intuitionistic fuzzy sets and soft sets, which offer different perspectives for the structures containing the uncertainties. In this case, the notion overlap needed to be extended to handle this types of objects. In this perspective, in  \cite{PAIVA} the authors generalized the notion of overlap to the context of lattices and introduced a weaker definition, called a quasi-overlap, that arises from the removal of the continuity condition. To this end, the main
properties of (quasi-) overlaps over bounded lattices, namely: convex sum, migrativity, homogeneity, idempotency, and cancellation law was investigated, as well as an overlap characterization of Archimedian overlap functions was presented.

In this paper, we propose a theoretical framework of order theory and topo\-logy with a view to establishing a connection between the notion of convergence in terms of order and Scott's topolo\-gy, to obtain a pair of residuated applications, namely: $ (O, I_O ) $, where $ O $ is a quasi-overlap function and $ I_O $ is an induced implication of $ O $. We proved that the class of quasi-overlap functions that fulfill the residuation principle is the same class of Scott-continuous quasi-overlap functions. Also, Scott's continuity and the notion of densely ordered posets are used to generalize a classification theorem for residuated quasi-overlap functions. Finally, the concept of automorphisms are extended to the context of quasi-overlap functions over lattices, taking these lattices into account as topolo\-gical spaces, with a view to obtaining quasi-overlap functions conjugated by the action of automorphisms.

To this end, the Section 2 presents an interaction between order theory and topology. The directed complete posets class (DCPO's) and the filtered complete posets class (FCPO's) as well as the lattice class are briefly explored. This section also shows  the notion of convergence via nets (a generalization of sequences for general topological spaces.). Finally, an overview of Galois connections and residuated mappings is recalled. In Section 3, we present some results investigated on residuated implications induced by fuzzy conjunctions that extend overlap functions to any lattice and some properties that these implications satisfy are presented. This section also shows how the notion of dense order coincides with the concept of density of topological spaces in Scott's topology. In section 4, presents a definition that generalizes automorphisms of bounded lattices, taking these lattices as topological spaces and the class of quasi-overlap functions  is closed under $ \Omega $-automorphisms, where $ \Omega $ represents, in this context, Scott's topology. In addition, some immediate properties of the action of Scott-automorphisms on quasi-overlap functions are explored. Finally, Section 5 gives some final remarks.

	\section{Preliminaries}

\subsection{Partial orders} 

In this subsection we will review some results of order theory, the branch of mathematics that deals among other things with order relations. For more details we recommend \cite{dickmann,gierz, Goubault, MISLOVE,scott2}.

\begin{definition}
	Let $\langle X, \leq \rangle$ be a poset. A subset $D$ of $X$ is called \emph{directed} if $D$ is not empty and $\forall u,v \in D$, $\exists w \in D$ such that $u\leq w$ and $v\leq w$. On the other hand, a subset $F$ of $X$ is called \emph{dual directed} or \emph{co-directed} or  \emph{filtered} if $F$ is not empty and $\forall u,v \in F$, $\exists w \in F$ such that $w\leq u$ and $w\leq v$.
\end{definition}

	\begin{remark}
	Since one usually can work on the dual order explicitly, notions of directed set and filtered set satisfy the principle of duality. 
\end{remark}

In what follows, it is easy to prove the

\begin{lemma}
	Let $\langle X, \leq \rangle$ be a poset. The following are valid:
\begin{enumerate}[label=(\roman*)]
	\item A non-empty chain in $X$ is directed and filtered;
	\item  For any $x \in X$, the set $\downarrow \!x=\{y \in X\, |\, y\leq x \}$ is directed and $\sup \downarrow \!x=x$;
	\item For any $x \in X$, the set $\uparrow \!x=\{y \in X\, |\, x\leq y \}$ is filtered and $\inf \uparrow \!x=x$;
	\item In a finite poset $X$, a subset of $X$ has maximal element  $\top$ if, and only if, it is directed;
	\item In a finite poset $X$, a subset of $X$ has minimal element $\bot$ if, and only if, it is filtered.
\end{enumerate}
\end{lemma}
\begin{remark}
	The sets $\downarrow \!x$ and $\uparrow \!x$ are known in the literature by principal ideal generated by $ x $ and principal filter generated by $ x $ respectively.
\end{remark}

\begin{definition}[\cite{dickmann}, p. 587]
	A poset $\langle X, \leq \rangle$ is called a complete partial order with respect to directed sets (\emph{DCPO}), if any directed subset of  $X$ has supremum in $X$. Dually, a poset $\langle X, \leq \rangle$  is called a complete partial order with respect to filtered sets (\emph{FCPO}), if any filtered subset of  $X$ has infimum in $X$. 
\end{definition}
\begin{remark}\label{QCPO}
	We remember that every poset $\langle X, \leq \rangle$ in which any two elements $x, y \in X $ have infimum and supremum, denoted respectively by $x \wedge y$ and $ x \vee y$, is called \emph{lattice}. We also remember that a lattice is said to be \emph{complete}, if for every non-empty subset $Y$ of $ X $, $Y$ has infimum and supremum in $ X $. Thus, every complete lattice is a DCPO and a FCPO.
\end{remark}

\begin{definition}[Order-density of posets] \label{order density}
	Let $ \langle X, \leq \rangle$ be a poset. A subset $Y$ of $X$ is order dense in $X$ if for any elements $ x, y \in X $ satisfying the condition $ x <y $, there exists an element $ z $ in $Y$ such that $ x <z <y$. If $X$ himself has order dense then $ \langle X, \leq \rangle$ has order dense.
\end{definition}
\subsection{Scott topology}

This subsection will discuss important issues for the development of this paper. It is assumed that the reader is familiar with some elementary notions of general topology, such as the notions of topological spaces, open and closed sets, the basis of a topology, as well as the separation axioms. Some of the results presented are well known in the literature, however, For more details we suggest \cite{gierz,Kelley,munkres_topology}.

\begin{definition}[Scott's open sets] \label{aberto}
	Let $\langle X, \leq \rangle$ be a DCPO and $A\subseteq X$. The set $A$ says a Scott open if it satisfies:\begin{enumerate}[label=(\roman*)]
		\item If $x\in A$ and $x\leq y$ then $y\in A$;
		\item If $D\subseteq X$ is a set directed and $\sup D \in A$ then $D\cap A \neq \emptyset$.
	\end{enumerate}
\end{definition}
\begin{proposition}[\cite{gierz}, Remark II-1.4]
	Let $\langle X, \leq \rangle$ be a DCPO and consider the set $$\sigma(X)=\{A\subseteq X\,|\, A \textrm{ is a Scott open set}\}.$$ Under these conditions, $\sigma(X)$ is a  topology on $X$. Moreover, $\langle X, \sigma(X) \rangle$ is a topological $T_0$\footnote{ A topological space $X$ is a $T_0$ space or Kolmogorov space if, for any two different points $x$ and $y$ there is an open set which contains one of these points and not the other.} space.
\end{proposition}
\begin{remark}
	$\sigma(X)$ it's well-known Scott's topology on $X$. 
\end{remark}

We recall that the notation $\langle X, \mathcal{T}, \leq \rangle$ is used to denote a set $X$ endowed with a topology $\mathcal{T}$ and a order relation ``$\leq$'' on $X$. Such a structure is called \textit{ordered topolo-gical space}. From now on, every DCPO will be considered an ordered topological space, in which the topology considered is Scott's topology.

\begin{proposition}[\cite{gierz}, Proposition II-2.1] \label{TD}
	Given a function $f:X \to Y$, where $X$ and $Y$ are DCPO's. The following conditions are equivalents: \begin{enumerate}[label=\textbf{(\roman*)}]
		\item $f$ is continuous with respect to Scott's topology: $f^{-1}(V)\in \sigma(X)$, for all $V\in \sigma(Y)$;\label{(i)}
		\item $f$ preserves supremum of directed sets, that is, $f$ preserves order and satisfies: $f(\sup D) = \sup f(D)$, for every directed set $D$ of $X$.
	\end{enumerate}
\end{proposition}

Similarly, the dual theorem can be demonstrated.

\begin{proposition}\label{TF}
	Given a function $f:X \to Y$, where $X$ and $Y$ are FCPO's. The following conditions are equivalents: \begin{enumerate}[label=\textbf{(\roman*)}]
		\item $f$ is continuous with respect to Scott's topology: $f^{-1}(V)\in \sigma(X)$, for all $V\in \sigma(Y)$;
		\item $f$ preserves infimum of filtered sets, that is, $f (\inf F) = \inf f(F)$, for every set filtered $F$ of $X$.
	\end{enumerate}
\end{proposition}

Therefore, considering the Remark \ref{QCPO} and the Propositions \ref{TD} and \ref{TF}, the next result is quickly obtained.

\begin{theorem}\label{conecta}
	Given a function $f:X \to Y$, where $X$ and $Y$ are complete lattices. The following conditions are equivalents: \begin{enumerate}[label=\textbf{(\roman*)}]
		\item $f$ is continuous with respect to Scott's topology: $f^{-1}(V)\in \sigma(X)$, for all $V\in \sigma(Y)$;\label{(i)}
		\item $f$ preserves supremum of directed sets and infimum of filtered sets: $f$ preserves order and satisfies: $f(\sup D) = \sup f(D)$ and $f (\inf F) = \inf f(F)$, for every directed set $D$ and every set filtered $F$, both subsets of $X$.
	\end{enumerate}
\end{theorem}

\subsection{Scott-convergence}

An interesting concept in topological spaces is that of convergence, as well as the properties of boundary points. In this section we will discuss these and other subject in terms of nets. For more details we suggest \cite{gierz,Kelley}.

\begin{definition}
	A net in a set $X$ is a function $j\mapsto x_j: J \to X$ whose domain is a join-semilattice \footnote{We recall that $J$ is a join-semilattice if $J$ is a partially ordered set that has a join (a least upper bound) for any nonempty finite subset of $J$.}. Nets are also denoted by $(x_j)_{j \in J}$, by $(x_j)$ or $x_j$, whenever the context is clear. If the set $X$ is provided with an order, then the net $x_j$ is called monotonic if $i\leq j$ implies $x_i \leq x_j$. A subnet of $(x_j)$ is any net of the form $(x_{\psi(i)})_{i \in I}$, where $I$ is a directed set  and there is an application $\psi:I \to J$ such that for each $j\in J$ there is eventually $\psi(i) \geq j$ in $I$.
\end{definition}

In the following definition, it is verified that the convergence of nets is a natural generalization of the convergence of sequences.

\begin{definition}\label{converg}
	A net $(x_j)_{j \in J}$ in a topological space $X$ converges to $x \in X$ (notation $x_j \to x$) if, whenever $U\subseteq X$ is open and $x \in U$, so  there is a $i \in J$ such that $x_j \in U$ for all $j \geq i$. 
\end{definition}

\begin{remark}
	Each subnet of a net that converges to a point (relative to a topological space) converges to the same limit. \cite[Affirmation b, p. 74]{Kelley}
\end{remark}

We still recall that given a topological space $X$, a collection $\mathcal{A}=\{A_{\lambda}\}_{\lambda \in I}$ of subsets of $X$ is called a covering of  $X$, when $X\subseteq \bigcup_{\lambda \in I} A_{\lambda}$. We say $\mathcal{A}$ is an open (closed) cover of $X$ when all elements of the cover are opened (closed). A topological space $X$ is called \textit{compact} when all open covering of $X$ has a finite subcollection that covers it. We say that $Y$ is a compact subset of $X$ if
$Y$, with the topology induced by $X$, is a compact topological space.

We now present a well-known characterization of topological space compactness via nets.

\begin{proposition}[\cite{Kelley}, Theorem 2, p. 136]\label{compa}
	A topological space $X$ is compact if and only if each net in $X$ admits a subnet converging to a point of $X$.
\end{proposition}

In general, a net in a topological space $ X $ can converge to several different points. For example, consider the two element set $\{a,b\}$ with topology \linebreak $\{\emptyset,\{b\},\{a,b\}\}$. Then every net, that converges to $a$ also converges to $b$ and the net, which is constant $b$ converges only to $b$. However, the following proposition points out spaces in which the convergence is unique in the sense that if a net $ s_n $ converges to $ s $ and also to a point $ t $, then $ s = t $. Before, we recall some useful notions. By a neighbourhood of a subset $A$ (in particular of a singleton and therefore of a point) in a topological
space $X$, we mean a subset of $X$ that contains an open set containing $A$. We say that a topological space $X$ is \textit{Hausdorff} if  any two distinct points of $X$ have disjoint neighbourhoods.

\begin{proposition}[\cite{Kelley}, Theorem 3, p. 67] \label{2.18}
	A topological space $X$ is a Hausdorff space if and only if each net in $X$ converges to at most one point.
\end{proposition}

In view of proposition \ref{2.18} above, the notion of lower bound and upper bound for nets is given below. This is a particular case of Definition II-1.1 in \cite{gierz}, for the case where $X$ is a complete lattice.

\begin{definition} \label{def_lim1}
	Given a complete lattice $X$ and a net $(x_i)_{i \in J}$ in $X$, the lower limit of $(x_i)_{i \in J}$ is:
	
	\begin{equation}
	\underline{\lim}_{i\in J} x_i= \sup_{i\in J} \inf_{j \geq i} x_j 
	\end{equation}
	and its upper limit is:
	
	\begin{equation}
	\overline{\lim}_{i\in J} x_i= \inf_{i\in J} \sup_{j \geq i} x_j 
	\end{equation} Let $S$ be the class of those elements $u \in X$ such that $u \leq \underline{\lim}_{i\in J} x_i$ and $T$ be the class of those elements $w \in X$ such that $\overline{\lim}_{i\in J} x_i\leq w$. For each such elements we say that $u$ is a lower $S$-limit and $w$ is a upper $T$-limit of $(x_i)_{i \in J}$. In this case we write respectively  $u \equiv_S \underline{\lim}_{i\in J}\ x_i$ and $w \equiv_T \overline{\lim}_{i\in J}\ x_i$. 
\end{definition}

%
%
%

\begin{proposition} [\cite{gierz}, prop. II-2.1] \label{lat_cont}
	Let $X$ and $Y$ be DCPO's and $f:X \rightarrow Y$  a function. The following conditions are equivalent:
	\begin{enumerate}
		\item $f$ preserves suprema of directed sets, i.e. $f$ is order preserving and 
		\begin{equation} \label{scott_cont}
		f(\sup \Delta) = \sup \{f(x) \ | \ x \in \Delta\}
		\end{equation}
		for all directed subset $\Delta$ of $X$,
		\item $f$ is order preserving and 
		\begin{equation}\label{88}
		f(\underline{\lim}_{i\in J} x_i) \leq \underline{\lim}_{i\in J} f(x_i)
		\end{equation}
		for any net $(x_i)_{i \in J}$ on $X$ such that $\underline{\lim}_{i\in J} x_i$ and $\underline{\lim}_{i\in J} f(x_i)$ both exist.
	\end{enumerate}
\end{proposition}

Similarly, the dual proposition can be demonstrated.

\begin{proposition}\label{lat_cont dual}
	Let $X$ and $Y$ be FCPO's and $f:X \rightarrow Y$  a function. The following conditions are equivalent:
	\begin{enumerate}
		\item $f$ preserves infimum of filtered sets, i.e. $f$ is order preserving and 
		\begin{equation} \label{scott_cont dual}
		f(\inf \Delta) = \inf \{f(x) \ | \ x \in \Delta\}
		\end{equation}
		for all filtered subset $\Delta$ of $X$;
		\item $f$ is order preserving and 
		\begin{equation}
		f(\overline{\lim}_{i\in J} x_i) \geq \overline{\lim}_{i\in J} f(x_i)\label{eq10}
		\end{equation}
		for any net $(x_i)_{i \in J}$ on $X$ such that $\overline{\lim}_{i\in J} x_i$ and $\overline{\lim}_{i\in J} f(x_i)$ both exist.
	\end{enumerate}
\end{proposition}

Notice that all complete lattice is a DCPO (FCPO) in which $\underline{\lim}_{i\in J} x_i$ and $\underline{\lim}_{i\in J} f(x_i)$ ($\overline{\lim}_{i\in J} x_i$ and $\overline{\lim}_{i\in J} f(x_i)$) always exist \cite{gierz}. Hence Propositions \ref{lat_cont} and \ref{lat_cont dual} hold for complete lattices. 

Theorem \ref{conecta} establishes a connection between convergence given in terms of lower bound order theory, or liminfs, and Scott's topology. In this perspective, Equations (\ref{scott_cont}) and (\ref{scott_cont dual}) generalize the notion of  continuity of functions on lattices. These facts motivate the following definition.

\begin{definition}\label{continuous}
	Let $X$ and $Y$ be two complete lattices. A function $f:X \rightarrow Y$ is Scott-continuous (simply continuous, if the context is clear) if and only if it satisfies any of the Equations (\ref{scott_cont}) or (\ref{scott_cont dual}).
\end{definition}

\begin{remark} \label{rem-cont-fin-L} Note that if $X$ is finite so any function $f:X \rightarrow Y$ is continuous because for each directed set $\Delta$ of $X$, $\sup \Delta\in \Delta$ and for each filtered set $\Delta$, $\inf \Delta\in \Delta$.
\end{remark}
\subsection{Galois connections and the residuation principle}

Galois connections generalize the correspondence between subgroups and \linebreak fields investigated in Galois theory. In order theory, a Galois connection is a particular correspondence between posets and is closely related to the concept of residuated functions. In turn, the residuated functions, besides being important in themselves, have a very relevant role in the characterization of partial orders. In this section we provide the basics necessary for the development of this paper. More details at \cite{Birkhoff, blyth,davey,Galatos,ore}.

\begin{definition}
A monotonic Galois connection from a poset $X$  to a poset $Y$ is a pair $(\alpha,\beta)$ of monotonic applications $X\stackrel{\alpha}{\to}Y\stackrel{\beta}{\to}X$ such that for all $x \in X$ and $y \in Y$, one has that $$\alpha(x) \leq_Y y \Longleftrightarrow x \leq_X \beta (y).$$ The $\alpha$ application is called the lower adjunct while the $\beta$ application is called the upper adjunct connection: 	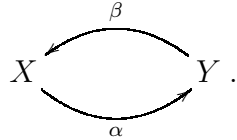
\begin{figure}[!htb]
	$$\xymatrix{X       \ar@/_0.6cm/[rr]_{\alpha}&&\ar@/_0.6cm/[ll]_{\beta}Y}.$$
	\caption{Galois monotonic connection between the posets $X$ and $Y$}\label{monoCG}
\end{figure}
\end{definition}

In the family of monotonic functions between partial orders there is a very important class of functions, called residuated functions.

\begin{definition}
	A function $f:X \to Y$ between posets $ X $ and $ Y $ is said to be residuated if it is monotonic, and in addition there is a monotonic function $ g: Y \to X $  such that \begin{equation}\label{eq3.3}
	g \circ f \geq \textrm{id}_X \quad \textrm{ and  }  \quad f \circ g \leq \textrm{id}_Y.
	\end{equation}
\end{definition}

If $f$ is a residuated function, the monotonic function $g$ that satisfies the inequalities in (\ref{eq3.3}) is called \textit{residue} of $f$ and is denoted by $f^R$. It is easy to see that the residue of a residuated function, when it exists, is unique (\cite{blyth}, p. 7).

Residuated functions, besides being important in themselves, play a fundamental role in the preservation of ideals.
	
\begin{theorem}[\cite{blyth}, Theorem 1.3] \label{preserva ideal}
	Let $f:X \to Y$ be an application between posets. The following statements are equivalents: \begin{enumerate}[label=\textbf{(\roman*)}]
		\item  $f$ is residuated;
		\item  \label{(ii)} For each principal ideal $\downarrow w$ of $Y$, the set $f^{-1}\left(\downarrow w\right)$ is a principal ideal of $X$.
	\end{enumerate}
\end{theorem}

\begin{theorem}[\cite{Galatos}, Lemma 3.2]
	Let $X$ and $Y$ be posets. An application \linebreak $f:X \to Y$ is residuated if, and only if, the pair  $(f,f^R)$ forms a monotonic Galois connection.
\end{theorem}


Therefore, whenever $ f $ is a residuated function, it is established that the pair $(f, f^R)$ satisfies the residuation principle, or adjunct principle, or that it forms a Galois connection. The next result can  be found in \cite{davey}, page 162, in topic 7.33. It is a characterization theorem for residuated functions.

\begin{theorem}[Characterization of residuated functions]\label{Caracteriza  Residuadas}
	Let $f:X \to Y$ and \linebreak $g:Y \to X$ be functions any between the posets $X$ and $Y$. The following statements are equivalents: \begin{enumerate} [label=\textbf{(\roman*)}]
		\item  $f$ is residuated and $g=f^R$; \label{R1}
		\item  For all $x \in X$ and $y \in Y$ one has $x \leq g(y) \Leftrightarrow f(x) \leq y$;\label{R2}
		\item  $f$  is monotonic and for each $y \in Y$, $g(y)=\max \{x\in X\,|\, f(x)\leq y\}$;\label{R3}
		\item  $g$  is monotonic and for each $x \in X$, $f(x)=\min \{y \in Y\,|\, x \leq g(y)\}$.\label{R4}
	\end{enumerate}
\end{theorem}

Galois connections are used to describe classes of functions for modeling fuzzy logic connectors. In the next section, the notions of quasi-overlap on lattices, as well as their derived implications, will be investigated.

\section{Residuated implications derived from quasi-overlap on lattices}

	In this section, we present some results investigated on residuated implications induced by conjunctions that extend overlap functions to any lattice. Overlap functions were proposed by Bustince et al. \cite{BUSTINCE} in order to  solve the problem of fuzziness on the process of image classification. Initially, overlap functions were defined as continuous functions.  Bustince et al. in \cite{BUSTINCE} justify the requirement of continuity by saying that it is considered in order to avoid $O$ to be a uninorm, however it is easy to see that if a uninorm is an overlap function, then it is necessarily a t-norm. However,	in some contexts, continuity is not an indispensable property, especially when we consider finite lattices. This situation appears in some situations in the field of digital image processing. Considering this, in \cite{PAIVA} the authors proposed  a more general definition, called of \textit{quasi-overlap}, which arises from the removal of the continuity condition.

\begin{definition}[\cite{PAIVA}, Definition 3.2]\label{overlaps}
	Let $X$ be a bounded lattice. A function $O:X^2 \rightarrow X$ is called a quasi-overlap  function on $X$ (simply quasi-overlap, if the context is clear) if all of following properties hold:
	\begin{enumerate}[labelindent=\parindent, leftmargin=*,label=\textbf{(OL\arabic*)}]
		\item $O(x,y)=O(y,x)$ for all $x, y \in X$;\label{ol1}
		\item $O(x,y)=0$ if and only if $x=0$ or $y=0$;\label{ol2}
		\item $O(x,y)=1$ if and only if $x=y=1$;\label{ol3}
		\item $O$ is non-decreasing in each variable, that is \begin{eqnarray*}
			x_1\leq x_2 \Rightarrow O(x_1,y)\leq O(x_2,y)\\
			y_1\leq y_2 \Rightarrow O(x,y_1)\leq O(x,y_2). \label{ol4}
		\end{eqnarray*}
	\end{enumerate}\end{definition}
\begin{proposition}\label{salvou}
	A quasi-overlap $O$ is associative if and only if, for any $x,y,z \in X$, it satisfies the exchange principle: $O\left(x,O(y,z)\right)=O\left(y,O(x,z)\right)$.
\end{proposition}

\begin{proof}
	The necessary and the sufficient conditions follows easily from the commutativity of $O$.
\end{proof}

In order to extend the notion of continuity presented in the Definition \ref{continuous} for the context of quasi-overlap functions, the following Definition is considered.

\begin{definition}[\cite{PAIVA}, Definition 3.1]
	Let $X$ be a complete lattice. An Overlap function on $X$ is a quasi-overlap that is Scott-continuous.
\end{definition}

In the following, the concept of residuated implications derived from quasi-overlap functions on lattices will be introduced.

\begin{definition}
	Let $X$ be a bounded lattice. A binary operation $I:X^2 \to X$ is called a implication if it is descending on the first variable, and nondecrea\-sing with respect to the second variable.  Moreover, $I(0,0)=I(0,1)=I(1,1)=1$ and $I(1,0)=0$.
\end{definition}

In the following, some properties that implications satisfy are presented.

\begin{definition}
	An implication $ I $ is said to fulfill: \begin{itemize}
		\item [\textbf{(NP)}] Neutral Property: $I(1,y)= y, \,\textrm{where }\, y \in X$;\label{PN}
		\item [\textbf{(EP)}] Exchange Principle: $I(x,I(y,z))=I(y,I(x,z))\,\textrm{where } \, x,y,z \in X$; \label{PT}
		\item [\textbf{(IP)}] Identity Principle:  $I(x,x)=1\,\textrm{where }\,x\in X$; \label{PI}
		\item [\textbf{(OP)}] Ordering Property: $x\leq y \Leftrightarrow I(x,y)=1\,\textrm{where }\,x,y \in X$. \label{OP}
	\end{itemize}
\end{definition}

\begin{lemma}\label{l5.1}
	Let $X$ be a complete lattice. Given a quasi-overlap $O:X^2 \to X$, the function defined by formula \begin{equation}\label{supL}
	I_O(x,y)=\sup \{t\in X\,|\, O(x,t)\leq y\}, \, \forall x,y \in X
	\end{equation} is nondecreasing with respect to second variable and decreasing with respect to first variable. Moreover, $I_O(0,0)=I_O(0,1)=I_O(1,1)=1$ and $I_O(1,0)=0$.
\end{lemma}
\begin{proof}
	In fact, first note that the function $I_O$ is well defined. Fix $x,y \in X$ and denote \begin{equation}\label{e5.2}
	R(x,y)\coloneqq \{t\in X\,|\, O(x,t)\leq y\}.
	\end{equation} Since $0 \leq O(x,0)=0$, it follows that $0 \in R(x,y)$. This means that $R(x,y) \neq \emptyset$ and once $X$ is complete, there is $\sup R(x,y)$ in $X$. Let $x,y,v \in X$,  with $y \leq v$. Then $$\{t\in X\,|\, O(x,t)\leq y\}\subseteq \{t\in X\,|\, O(x,t)\leq v\}$$ and therefore, $\sup \{t\in X\,|\, O(x,t)\leq y\}\leq \sup\{t\in X\,|\, O(x,t)\leq v\}$,that is, $I_O(x,y)\leq I_O(x,v)$. This means that the function $I_O$ is nondecreasing on the second variable. Now, let $x,u,y \in X$, com $x\leq u$. From the monotonicity of $ O $ with respect to the first variable one has that $O(x,t) \leq O(u,t)$, for all $t\in X$. Therefore, $$\{t\in X\,|\, O(u,t)\leq y\}\subseteq \{t\in X\,|\, O(x,t)\leq y\},$$ then $\sup\{t\in X\,|\, O(u,t)\leq y\}\leq \sup\{t\in X\,|\, O(x,t)\leq y\}$ and therefore, $I_O(u,y)\leq I_O(x,y)$. Thus $I_O$ is decreasing on the first variable. Moreover, since every quasi-overlap satisfies \ref{ol2} and \ref{ol3} from Definition \ref{overlaps}, one has that \begin{itemize}
		\item $I_O(0,0)=\sup \{t\in X\,|\, O(0,t)\leq 0\}=1$;
		\item $I_O(0,1)=\sup \{t\in X\,|\, O(0,t)\leq 1\}=1$;
		\item $I_O(1,1)=\sup \{t\in X\,|\, O(1,t)\leq 1\}=1$;
		\item $I_O(1,0)=\sup \{t\in X\,|\, O(1,t)\leq 0\}=0$.\end{itemize}
\end{proof}

\begin{definition}
	Let $O$ be a quasi-over a complete lattice $X$. The function $I_O$ defined by (\ref{supL}) is called the implication induced by $O$.
\end{definition}

In order to develop the notion of residuation for quasi-overlap and its induced implications, it is necessary to explore some important facts about a particular class of lattices, namely, the class of dense complete lattices. 
For this, we present below a sequence of useful results.

First we remember that given a subset $ B \subseteq X $ of a topological space $ X $, the interior of $B$, denoted by $\textrm{int}(B)$, is the largest open set contained in $ B $.

\begin{lemma}\label{i}
	Le $\langle X, \leq \rangle$ be a poset. If $B\subseteq X$, then $\textrm{int}(B)=\{x \in B \,|\, \uparrow \!x \subseteq B\}$.
\end{lemma}
\begin{proof}
	Let $w \in \{x \in B \,|\, \uparrow \!x \subseteq B\}$. Then $\uparrow\!w \subseteq B$. But $\uparrow\!w$ is open, hence $\uparrow\!w \subseteq \textrm{int}(B)$. Thus, $w \in\textrm{int}(B)$. On the other hand, let $y \in  \textrm{int}(B) \subseteq B$. Since $\textrm{int}(B)=\uparrow\!\textrm{int}(B)$, one has that $\textrm{int}(B)\subseteq \uparrow \!\!B$. Then, $y \in \{x \in B \,|\, \uparrow \!x \subseteq B\}$.
\end{proof}

We recall also that a subset $S$ of a topological space $X$ is \textit{dense} in $ X $ when its closure $ \overline{S}$ coincides with the whole space $ X $. This is equivalent to say that every open non-empty in $ X $ contains some point of $ S $, or else, that the complement of $ S $ does not have interior points. In order to establish a connection between this topological definition for density and that given by the Definition \ref{order density} and which is given in terms of order, we present the following assertion.

\begin{theorem}\label{densety}
	Let $\langle X, \leq \rangle$ be a DCPO. A set $S \subseteq X$ is order dense in $X$ if, and only if, $S$ is dense in $X$ in the Scott's topology.
\end{theorem}

\begin{proof}
$(\Rightarrow)$ Suppose that the set $S \subseteq X$ is order dense in $X$ and let $x,y \in X$ be such that $x<y$. Then, the set $D=\{z \in S\,|\, x < z < y\}$ is non-empty and directed, with $\sup D = y$. On the other hand, $\uparrow \!x=\{r\in X \,|\, x\leq r\}$ is an open set of $X$ in Scott's topology, since it satisfies the properties $(i)$ and $(ii)$ of Definition \ref{aberto}. Therefore, since $\uparrow \!x \cap D\neq \emptyset$ and $D\subseteq S$, it follows that $\uparrow \!x \cap S \neq \emptyset$. Therefore, $S$ is dense in $X$  in the Scott's topology.\\
	
In that follows, if neither $x \leq y$ nor $y \leq x$, then $x$ and $y$ are said to be incomparable, which is denoted here as $x\parallel y$.\\
	
	\noindent
$(\Leftarrow)$  For each $x \in X$ which satisfies the condition: $x < y$ or $x\parallel y$, for some $y\in X$, define the set $P_x = \{w \in X\,|\, x < w\textrm{ or } x\parallel y\}$. It is clear that $P_x \neq \emptyset$. Let's show that $P_x$ is a Scott's open. In fact, the condition ($i$) of Definition \ref{aberto} is trivially satisfied. As for condition ($ii$), suppose there is a directed set $D\subseteq X$ such that $\sup D \in P_x$ and $D\cap P_x = \emptyset$. So, for all $w \in D$ one has that $w \leq x$. Therefore, $x$ is an upper bound for $w$ and $\sup D \leq x$, which contradicts the fact that $\sup D \in P_x$. Therefore, $D\cap P_x \neq \emptyset$ and hence $P_x$ is a Scott's open. Moreover, since $S$ is topologically dense in $X$, one has that $P_x \cap S \neq \emptyset$ and hence exists $z \in S$ such that $x<z<y$ for all $x,y \in X$ that satisfy the condition $x<y$.
\end{proof}

\begin{proposition}	\label{filtro compacto}
	Let  $\langle X, \leq \rangle$ be a poset provided with Scott's topology $\sigma(X)$.  For all $y \in X$, the set $\uparrow \!y$ is compact. 
\end{proposition}
\begin{proof}
In fact, let $\mathcal{A}=\{A_{\lambda}\}_{\lambda \in I}$ be an open covering of $X$. Then, for all $y \in X$, $y \in \bigcup_{\lambda \in I} A_{\lambda}$. Hence, $y \in A_{\lambda_{0}}$, for some $\lambda_{0} \in I$. But $ A_{\lambda_{0}}$ is Scott's open, then $ A_{\lambda_{0}}=\uparrow \!A_{\lambda_{0}}$. Also, as $\uparrow \!y \subseteq \uparrow \!A_{\lambda_{0}}$, it follows that  $\uparrow \!A_{\lambda_{0}}$ is a finite subcolection of  $\mathcal{A}$. Therefore, since $\uparrow \!y \subseteq \uparrow \!A_{\lambda_{0}}$, it follows that $\uparrow \!y$ is compact.
\end{proof}

In Scott's topology, for complete lattices, compactness is a trivial feature.

\begin{proposition}\label{com}
	If $\langle X,  \sigma(X), \leq \rangle$ is a complete lattice equipped with Scott topo\-logy, then $X$ is compact.
\end{proposition}

\begin{proof}
let $\mathcal{A}=\{A_{\lambda}\}_{\lambda \in I}$ be an open covering of $X$. Since $X$ is complete  it follows that it is bounded. Let $0$ be its bottom element. Then, $0 \in \bigcup_{\lambda \in I} A_{\lambda}$. Hence, $0 \in A_{\lambda_{0}}$, for some $\lambda_{0} \in I$. But $ A_{\lambda_{0}}$ is Scott's open, then $ A_{\lambda_{0}}=\uparrow \!A_{\lambda_{0}}$. Also, as $\uparrow \!0 \subseteq \uparrow \!A_{\lambda_{0}}$, it follows that  $\uparrow \!A_{\lambda_{0}}$ is a finite subcolection of  $\mathcal{A}$. Also, since $\uparrow \!0 = \{x\in X\,|\, 0\leq x\}$, it follows that $X \subseteq \uparrow \!0$. Therefore, $X$ is compact.
\end{proof}

Therefore, follows from Theorem \ref{densety} and Proposition \ref{com},  which every complete dense order lattice is compact and dense in Scott's topology.

Another point that deserves attention is the fact that these spaces must not necessarily have a definite total order about them. Consider the set of subintervals of $[0,1]$ defined as $\mathbb{I}([0,1])= \left\{[a,b]\,|\, a\leq b \textrm{ and } a,b \in [0,1] \right\}$, provided with the product order ``$ \preceq $'' defined as follows: $$[u,v]\preceq [p,q] \textrm{ if, and only if, } u \leq p \textrm{ and } v \leq q$$ where ``$\leq$'' is the usual order of $\mathbb{R}$. The next lemma shows that $\langle \mathbb{I}([0,1]) , \preceq \rangle$ is a partially ordered complete lattice.

\begin{lemma}
	The structure $\left\langle \mathbb{I}([0,1]), \sqcap, \sqcup, \mathbf{0}, \mathbf{1} \right\rangle$, where $[a,b] \sqcap [c,d] = [a\wedge c, b \wedge d]$ and $[a,b] \sqcup [c,d] = [a\vee c, b \vee d]$, for all $a,b,c,d \in [0,1]$, it is a complete lattice with the top element $\mathbf{1}=[1,1]$ and the bottom element $\mathbf{0}=[0,0]$.
\end{lemma}
 \begin{proof}        
	According to the definition of $\sqcap$ and $\sqcup$ operators, just consider for each $[a,b],[c,d] \in \mathbb{I}([0,1])$, 
	$\inf \{[a,b], [c,d]\}=[a,b]\sqcap [c,d]$ and $\sup \{[a,b], [c,d]\}=[a,b]\sqcup [c,d]$. Thus 
	$\langle \mathbb{I}([0,1]), \sqcap, \sqcup,\mathbf{0}, \mathbf{1} \rangle$ is a lattice. Now consider the non-empty set  $X \subseteq \mathbb{I}([0,1])$. It is obvious that $[0,0]$ 
	is a lower bound of $ X $, then the set: $$X^{\ell}=\{J \in \mathbb{I}([0,1])\,|\, J \textrm{ is lower bound of } X\}$$ it is not empty. Define 
	$$v=\sup\limits_{[p,q]\in X^{\ell}} (p) \quad \textrm{ and } \quad w=\sup\limits_{[p,q]\in X^{\ell}}(q).$$ 
	This implies that $[v,w]$ is lower bound of $X$. We affirm that  $[v,w]$ it is the largest of the lower bounds of $ X $. Indeed, suppose there 
	exists $[r,s] \in X^{\ell}$ such that $[v,w]\sqsubseteq [r,s]$, then $v\leq r \, \textrm{ and }\, w\leq s$. On the other hand, by way $v$ and $w$ 
	are defined, we have $v \geq r$ and $w\geq s$. Therefore $v=r$, $w=s$ and hence $\inf X = [v,w]$. Similarly, since $[1,1]$ is upper bound 
	for $X$, define $$X^{u}=\{T\in \mathbb{I}([0,1])\,|\, T \textrm{ is upper bound of } X\}\neq \emptyset.$$ and call 
	$$m=\inf\limits_{[a,b]\in X^{u}}(a) \quad \textrm{ and } \quad n=\inf\limits_{[a,b]\in X^{u}} (b).$$ 
	This leads us to conclude that $\sup X = [m,n]$.     
\end{proof}

 In addition, since $ [0,1] $ is dense, it follows that $ \mathbb{I}([0,1]) $ is also dense. On the other hand, the interval $ [0,1] $ provided with the usual order of reais is also complete lattice of dense order. Finally, a final point to be discussed on this topic is linked to the issue of convergence, which is clarified in the next lemma.

\begin{lemma}\label{tt}
	If $X$ is a complete lattice of dense order, then $X$ is Hausdorff.
\end{lemma}
\begin{proof}
	Let $x,y \in X$, so we have the following possibilities: \begin{enumerate}[labelindent=\parindent, leftmargin=*,label=\textbf{(\roman*)}]
		\item $x$ and $y$ incomparable: In this case, $x,y < \sup\{x,y\}$. Then, of the density of $X$,  there are $x_0,y_0\in X$ such that $x< x_0 < \sup\{x,y\}$ and $y< y_0 < \sup\{x,y\}$. Let's show that the sets $B=\{w\in X\,|\, x\leq w < x_0\}$ and $C=\{u\in X\,|\, y\leq u < y_0\}$ are open non-empty of $X$  in Scott's topology and that $B\cap C=\emptyset$. Indeed, first note that by definition $\textrm{int}(B)\subseteq B$ always worth it. On the other hand, if $b\in B\subseteq X$, then $x\leq b < x_0$ and more, the set $\uparrow\!x \cap B= \{r\in B\,|\, x\leq r\}$ is an open of $B$ in Scott's topology induced of $X$. Therefore, by Lemma \ref{i}, $b\in \textrm{int}(B)=\{r\in B\,|\, \uparrow\!r \subseteq B\}$. That is, $B$ is an open of $X$ that contains $x$. Similarly it is shown that  $C=\textrm{int}(C)$. That is, $C$ is a open of $X$ that contains $y$. Moreover, if $B\cap C \neq \emptyset$ so there is $\alpha \in X$ such that $x\leq \alpha$ and $y\leq \alpha$, which leads to a contradiction with $\sup\{x,y\}$.
		\item $x$ and $y$ are comparables: In this case, assume without loss of generality  $x< y$. From the density of $ X $ there is $z\in X$ such that $x< z< y$. Define the sets $M=\{p\in X\,|\, x\leq p< z\}$ and $\uparrow \! z=\{q\in X\,|\, z\leq q\}$. Notice that $M$ is a open of $X$ that contains $x$ and $\uparrow \!z$ is a open of $X$ that contains $y$. Moreover, $M \cap \uparrow \!z = \emptyset$.
	\end{enumerate} Therefore, it follows that the lattice $ X $ is a Hausdorff space.

\end{proof}
\begin{remark}
	The Lemma \ref{tt} together with Proposition \ref{2.18} guarantee the uniqueness of convergence of converging nets in a complete lattice of dense order.
\end{remark}

Let $ X $ be the complete lattice of dense order. For each $ x \in X $, define the functions $O_x, I_{O_x}:X \to X$ por $O_x(z)=O(x,z)$ and $I_{O_x}(y)=I_O(x,y)$, for all $y,z \in X$, where $O$ is a quasi-overlap and $I_O$ its induced implication. In what follows, it will be presented under what conditions $O_x$ and $I_{O_x}$ represent a family of residuated functions and their respective family of residues.

\begin{definition}
	The pair $(O,I_O)$ is said to satisfy the residuation principle whenever \begin{equation}
	O(x,z)\leq y \Longleftrightarrow z\leq I_O(x,y), \, \, \forall x,y,z \in X.
	\end{equation} 
\end{definition} 

The next theorem reveals that the class of quasi-overlap functions that fulfill the residuation principle is the class of continuous functions according to Scott's topology.

\begin{theorem}\label{resigalois}
	Let $X$ be a complete lattice of dense order and $O$ a quasi-overlap over $X$. So the following items are equivalent:\begin{enumerate}[labelindent=\parindent, leftmargin=*,label=\textbf{(\roman*)}]
		\item $O$ is Scott-continuous;
		\item $O$ and $I_O$ satisfy the residuation principle;
		\item $I_O(x,y) = \max\{t\in X\,|\, O(x,t)\leq y\}$. \label{(iii)}
	\end{enumerate}
\end{theorem}
\begin{proof}
	(\ref{(i)} $\Rightarrow$ \ref{(ii)}): For any $x,y, z \in X$ suppose that $O(x,z)\leq y$. Then $z \in R(x,y)$ (cf. Equation (\ref{e5.2})). Hence, $z \leq \sup R(x,y)=I_O(x,y)$. Now assume that for $x,y,z \in X$ one has $z \leq I_O(x,y)$. If $z < I_O(x,y)$ then, since $X$ is dense order, there is $t_0 \in X$ such that $z < t_0 < I_O(x,y)$ and $ O(x,t_0) \leq y$. From the monotonicity of $O$ in each variable one has thate $O(x,z)\leq y$. On the other hand, if $z=I_O(x,y)$, so we have two possibilities: \begin{enumerate}[labelindent=\parindent, leftmargin=*,label=\textbf{(P\arabic*)}]
		\item $z \in R(x,y)$: In this case obviously that $O(x,z)\leq y$;
		\item $z \notin R(x,y)$: In this case, since $X$ is complete and of dense  order, then $X$ It is compact and dense in Scott's topology. Thus, by Proposition \ref{compa}, There is a non-decreasing net  $(z_j)_{j\in J}$ in $X$ such that $z_j < z$ and, since $O$ is monotonic in the second variable, from the residuation principle it follows that $O(x,z_j)\leq O(x,z)\leq y$, for all $j \in J$. Let's show that $z=\underline{\lim}_{j\in J} z_j$. In fact, let $ A $ be a Scott open containing $ z $. Since $\{z_j \in X\,|\,O(x,z_j)\leq y\}$ is directed (since  $(z_j)$ is non-decreasing) and $z=\sup\{z_j \in X\,|\,O(x,z_j)\leq y\}$, then by item \ref{(ii)} from Definition \ref{aberto} (Scott's open), it follows $\{z_j \in X\,|\,O(x,z_j)\leq y\}\cap A \neq \emptyset$. Therefore, for some $i \in J$, we have $x_j \in A$ for all $j \geq i$. Thus, by Definition \ref{converg}, we have $z_j \to z$. That is, $z=\underline{\lim}_{j\in J} z_j$. Finally, since $O$ is Scott-continuous, by Proposition \ref{lat_cont} \begin{equation}
		O(x,z)=O(x,\underline{\lim}_{j\in J} z_j) \leq \underline{\lim}_{j\in J} O(x,z_j)\leq y.
		\end{equation}
	\end{enumerate}
Therefore, anyway, one has that  $O(x,z) \leq y$.
	
	\noindent
	(\ref{(ii)} $\Rightarrow$ \ref{(iii)}): Assume that pair $(O,I_O)$ satisfies the residuation principle. So since $I_O(x,y) \leq I_O(x,y)$ for all $x,y \in X$, it follows that $O(x,I_O(x,y))\leq y$. This means that $I_O(x,y)\in R(x,y)$ and  $\sup R(x,y)=\max R(x,y)$.
	
	\noindent
	(\ref{(iii)} $\Rightarrow$ \ref{(i)}): Suppose that $I_O (x,y)=\max \{t \in X\,|\,O(x,t) \leq y\}$ for all $x,y \in X$. We must show that $O\bigg(x,\sup\{z_j\,|\, j\in J\}\bigg)=\sup\{ O(x,z_j)\,|\,j\in J\}$, for each $x \in X$ and for any non-decreasing net $(z_j)_{j\in J}$ in $X$. On the one hand, the monotonicity of $O$ and by definition of supremum it follows that \begin{equation} \label{1*} \sup\{ O(x,z_j)\,|\, j\in J\} \leq O\bigg(x,\sup\{z_j\,|\, j\in J\}\bigg)\end{equation} On the other hand, let $w= \sup\{ O(x,z_j)\,|\, j\in J\}$. Then $O(x,z_j) \leq w$ and so for all $j \in J$, $z_j \in \{t\in X\,|\,O(x,t) \leq w\}$ and consequently, $z_j \leq I_O(x,w)$ for all $j \in J$. Therefore, by monotonicity of $O$ one has  \begin{equation} \label{2*} O\bigg(x,\sup\{z_j\,|\, j\in J\}\bigg)\leq O(x,I_O(x,w))\leq w =\sup\{ O(x,z_j)\,|\, j\in J\}.\end{equation} Therefore, from inequalities (\ref{1*}) and (\ref{2*}), it is concluded that $O$ is continuous.\end{proof}

\begin{corollary}
If $O$ is a Quasi-overlap over $X$ and $X$ is order dense, then $O$ and $I_O$ satisfy residuation principle and $I_O(x,y)=\max\{z\,|\, O(x,z)\leq y\}$.
\end{corollary}

\begin{definition}
	The functions $O$ and $I_O$ are respectively called residuated quasi-overlap and residuated implication (or $R_O$-implication), if any of the items in Theorem \ref{resigalois} are checked.
\end{definition}

In the following, properties that $ R_O $-implications and their residuated quasi-overlap satisfy are presented.

\begin{proposition}\label{propop}
	Let $X$ be a complete lattice of order dense and $O$ a residuated quasi-overlap  over $X$. Then:\begin{enumerate}[labelindent=\parindent, leftmargin=*,label=\textbf{(\roman*)}]
		\item $I_O$ satisfies \textbf{(NP)} if, and only if, $1$ is neutral element of $O$;

		\item $I_O$ satisfies \textbf{(IP)} if, and only if, $O$ is deflationary: $$O(x,1)\leq x, \,\,x\in X\textrm{;}$$
		\item $I_O$ satisfies \textbf{(OP)} if, and only if, $O$ have neutral element $1$.
	\end{enumerate}
\end{proposition}

	\begin{proof} 
The proof is based on considerations similar to \cite{Krol}. But adapted to the lattice context. Indeed, \\
	\noindent
	\ref{(i)} ($\Rightarrow$) Suppose that for all $y \in X$ \begin{equation}\label{6}
	I_O (1,y)=\max \{t \in X\,|\,O(1,t) \leq y\}=y.
	\end{equation} So for an arbitrary $y \in X$ one has $O(1,y) \leq y$. If for some $y_0$ in $X$, one has $O(1,y_0) < y_0$, then by density of $X$  exists $z$ such that $z< y_0$ and $O(1,y_0) < z$. According to the residuation principle, $z<  y_0\leq I_O(1,z)$, which contradicts the equation (\ref{6}). \\
	\noindent
	($\Leftarrow$)  Suppose that $O(1,r)=r$, for all $r\in X$. Then \begin{eqnarray}
	I_O (1,y)&=&\max \{t \in X\,|\,O(1,t) \leq y\}\nonumber\\
	&=&\max \{t \in X\,|\,t \leq y\}\nonumber\\
	&=&y.
	\end{eqnarray}
	\noindent
	\ref{(ii)} Just note that for an arbitrary $x \in X$, we have $$I_O(x,x)=\max  \{t \in X\,|\,O(x,t) \leq_X x\}=1 \Leftrightarrow O(x,1) \leq x.$$
	
	\noindent
	\textbf{(iii)} ($\Rightarrow$) Suppose for each $x,y \in X$, such that $x$ and $ y$ are comparables, one has $x\leq y \Leftrightarrow I_O(x,y)=1$. Then $I_O(x,x)=\max  \{t \in X\,|\,O(x,t) \leq x\}=1$. This means that $O(x,1) \leq x$, for all $x \in X$. Moreover, by the monotonicity of $O$, $$I_O(x,O(x,1))=\max \{t \in X\,|\,O(x,t) \leq O(x,1)\}=1.$$ Thus, by \textbf{(OP)}, $x\leq O(x,1)$. So for an arbitrary $x\in X$, $O(x,1)=x$.\\
	\noindent
	($\Leftarrow$) Suppose $ O $ has neutral element $ 1$. If for  $x,y \in X$,  $$I_O (x,y)=\max \{t \in X\,|\,O(x,t) \leq y\}=1,$$ then we have $x=O(x,1)\leq y$. On the other hand, if for each $x,y \in X$, if $x \leq y$, so since $1$ is neutral element of $O$, one has $O(x,1)=x\leq y$ . Therefore, by residuation, it follows that $I_O(x,y)=1$.\end{proof}

\begin{proposition}
Let $X$ be a complete lattice of order dense and $O$ a residuated quasi-overlap  over $X$. Under these conditions: \begin{enumerate}[labelindent=\parindent, leftmargin=*,label=\textbf{(\roman*)}]
		\item If $I_O$ satisfies \textbf{(EP)}, $O(x,O(y,z))$ and $O(y,O(x,z))$ are comparables for all $x,y,z \in X$, then $O$ is associative;
		
		\item If $O$ is associative, then $I_O$ satisfies \textbf{(EP)}.
		\end{enumerate}
\end{proposition}

\begin{proof}\noindent
	\textbf{(i)} Assume that $ I_O $ fulfills the property of the exchange principle \textbf{(EP)}. Suppose that there are $x,y,z \in X$ such that $O(x,O(y,z))\neq O(O(x,y),z)$. Then, by Proposition \ref{salvou} it follows that $O(x,O(y,z))\neq O(y,O(x,z))$. Thus, by hipoteses, we can assume without loss of generality that $O(x,O(y,z))< O(y,O(x,z))$. Applying two times the residuation principle we get $$z< I_O(y,I_O(x,O(y,O(x,z)))).$$ Using the exchange principle we have $z<I_O(x,I_O(y,O(y,O(x,z))))$. Applying the residuation principle again two times we go back to $$O(y,O(x,z))< O(y,O(x,z)),$$ which is trivially a contradiction.\\
	\noindent
	\textbf{(ii)} Assume that $ O $ is associative. From residuation principle we have \begin{eqnarray*}
		I_O(x,I_O(y,z))&=&\max \{t \in X\,|\,O(x,t) \leq I_O(y,z)\}\\
		&=&\max \{t \in X\,|\,O(y,O(x,t)) \leq z\}\\
		&=&\max \{t \in X\,|\,O(O(y,x),t) \leq z\}\\
		&=&\max \{t \in X\,|\,O(O(x,y),t) \leq z\}\\
		&=&\max \{t \in X\,|\,O(x,O(y,t)) \leq z\}\\
		&=&\max \{t \in X\,|\,O(y,t) \leq I_O(x,z)\}\\
		&=& I_O(y,I_O(x,z)).
	\end{eqnarray*}
\end{proof}

\section{Quasi-overlap conjugated and their induced implications}

We begin this section by presenting a definition that generalizes automorphisms of bounded lattices, taking these lattices as topological spaces.

\begin{definition}\label{pro-aut-char-top} 
	Let $X$ be a bounded lattice and $\Omega$ a topology on $X$. A function $\rho: X\rightarrow X$ is a $\Omega$-automorphism if:
	\begin{enumerate}[label=\textbf{(\roman*)}]
		\item $\rho$ is bijective;
		\item $\rho$ is continuous according to topology $\Omega$;
		\item $x\leqslant y$ if, and only if,  $\rho(x)\leqslant \rho(y)$.
	\end{enumerate}
\end{definition}

\begin{remark}
	The set of all $\Omega$-automorphism of $X$ is denoted by $\textrm{Aut}_{\Omega}(X)$.
\end{remark}
\begin{lemma}
The set	$\textrm{Aut}_{\Omega}(X)$ of all $\Omega$-automorphism of a bounded lattice$X$ is a group under composition of mappings. 
\end{lemma}

\begin{proof}
It is routine to check this proof.	
\end{proof}

Note that Definition \ref{pro-aut-char-top} generalizes the definition of automorphisms over $ [0,1] $, which implies continuity in the Euclidean topology. In addition, since $ \rho $ is a continuous bijection whose inverse $ \rho ^ {- 1} $ is also continuous, it follows that $ \rho $ is an application known in topology as homeomorphism. It is an application that preserves the topological structure of your space. It should also be noted that $ \rho $ (as well as its inverse) can be seen as an order isomorphism.

Furthermore, the next result shows that the class of quasi-overlap functions  is closed under $ \Omega $-automorphisms, where $ \Omega $ represents, in this context, Scott's topo\-logy, and for this reason, instead of $ \Omega $-automorphism the term Scott-automorphism is used.

\begin{proposition}\label{conjugated}
	Let $O$ be a quasi-overlap function and $\rho$ a Scott-automorphism, both defined over a complete lattice $X$. Then, the conjugated of $O$, denoted by $O^\rho$, is also quasi-overlap function. Moreover, if $O$ is Scott-continuous,  $O^\rho$ is also Scott-continuous.
\end{proposition}
\begin{proof}
	\ref{ol1}: It follows directly from the fact that composed of non-decreasing functions is a non-decreasing function.;\\	
	\ref{ol2}: Immediately follows from the commutativity of $O$;\\
	\ref{ol3}: ($\Rightarrow$) Suppose that $O^{\rho}(x,y)=0$. Then we have \begin{eqnarray*}
		\rho^{-1}\left(O\left(\rho(x),\rho(y)\right)\right)=0 &\Leftrightarrow& O\left(\rho(x),\rho(y)\right)=0\\
		&\Leftrightarrow& \rho(x)=0 \textrm{ or } \rho(y)=0\\
		&\Leftrightarrow& x=0 \textrm{ or } y=0
	\end{eqnarray*}
	\noindent
	($\Leftarrow$) If $x=0$ or $y=0$ then, suppose without loss of generality that $x = 0$. Then, $\rho(x)=0$ and so, \begin{eqnarray*}
		O^{\rho}(x,y)&=&	\rho^{-1}\left(O\left(\rho(x),\rho(y)\right)\right)\\
		&=&\rho^{-1}\left(O\left(0,\rho(y)\right)\right)\\
		&=&\rho^{-1}(0)\\
		&=&0.
	\end{eqnarray*}
	\noindent	
	\ref{ol4}: ($\Rightarrow$) Suppose that $O^{\rho}(x,y)=1$. Then one has that \begin{eqnarray*}
		\rho^{-1}\left(O\left(\rho(x),\rho(y)\right)\right)=1 &\Leftrightarrow& O\left(\rho(x),\rho(y)\right)=1\\
		&\Leftrightarrow& \rho(x)= \rho(y)=1\\
		&\Leftrightarrow& x= y=1
	\end{eqnarray*}
	\noindent
	($\Leftarrow$) Suppose that $x=y=1$, $\rho(x)=\rho(y)=1$. Then one has that \begin{eqnarray*}
		O^{\rho}(x,y)&=&\rho^{-1}\left(O\left(1,1\right)\right)\\
		&=&\rho^{-1}(1)\\
		&=&1.
	\end{eqnarray*}
	\noindent
	The continuity of $O^{\rho}$ follows from the composition of continuous functions. \end{proof}

A first application of the action of Scott-automorphisms on quasi-overlap functions is that the conjugated of an induced implication of a quasi-overlap $O$ coincides with the induced implication of the conjugated $O^{\rho}$.

\begin{proposition}\label{coincide}
	$I_O ^{\rho}$ coincides with $I_{O^{\rho}}$.
\end{proposition}
\begin{proof}
	Indeed, \begin{eqnarray*}
		I_O ^{\rho}(x,y) &=& \rho^{-1} \left(I_O(\rho(x),\rho(y))\right)\\
		&=&\rho^{-1} \left(\sup\{\rho(z) \in X| O(\rho(x),\rho(z)) \leq \rho(y)\}\right)\\
		&=&\max\{z \in X | \rho^{-1} \left(O(\rho(x),\rho(z))\right)\leq y\}\\
		&=& \max\{z \in X | O^{\rho} \left(x, z\right)\leq y\}\\
		&=&I_{O^{\rho}}(x,y).
	\end{eqnarray*}
\end{proof}

\begin{remark}
	The above proposition states that the processes for obtaining conjugated adjunct or adjunct conjugated are invariant, as shown in Figure \ref{diagrama1}.
	
	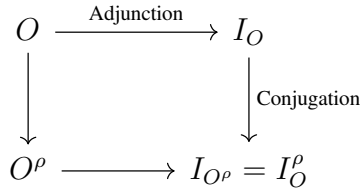
\begin{figure}[!htb]
		$$
		\begin{tikzcd}[row sep=large, column sep=large]
		O \arrow{r}{\textrm{Adjunction}} \arrow{d}
		& I_O \arrow{d}{\textrm{Conjugation}} \\
		O^{\rho} \arrow[r]	& I_{O^{\rho}}=I_O ^{\rho}
		\end{tikzcd}$$
		\caption{Diagram of Adjunct and Conjugation}\label{diagrama1}
	\end{figure}
\end{remark}

Another interesting application of quasi-overlap conjugated is linked to the notion of closed operators\footnote{Remember that a function $f: X \rightarrow X$ over a poset $\langle X, \leq \rangle$ is a closed operator on $X$ if $f$ is non-descreasing, idempotent $\left(f(f(x))=f(x)\right)$, and inflationary $\left(x \leq f(x)\right)$.}. The following is a theoretical framework for obtaining the closure of the conjugated of  $O$ and $I_O$.

\begin{proposition}\label{ou}
	Let $X$ be a complete lattice of dense order and $O^{\rho}$ a conjugated of quasi-overlap function  $O$ set over $X$. The following conditions are equivalent:
	
	\begin{enumerate}[labelindent=\parindent, leftmargin=*,label=\textbf{(\roman*)}]
		\item  $O^{\rho}$ is residuated if, and only if, $O_y ^{\rho}$ and $O_x ^{\rho}$ are both monotonic and Scott-continuous;
		
		\item $O^{\rho}$ is residuated if, and only if, $O_y ^{\rho}$ and $O_x ^{\rho}$ are residuated.
	\end{enumerate}
\end{proposition}

\begin{proof}
	(\ref{(i)} $\Rightarrow$ \ref{(ii)})  On the space $X$ define the following partial order relation: $$(a,b)\leq(u,v) \Leftrightarrow a\leq u \textrm{ and } b\leq v.$$
	
It is routine to check that this provides two natural topologies $X^2$, namely: Scott's topology in space $\langle X, \leq \rangle$ and Scott's topology product in $\langle X^2, \leq \rangle$. So if a function defined on $X^2$ is Scott-continuous, its projections on the $X$ factor are clearly Scott-continuous. Therefore, since $O$ is commutative, non-decreasing at each variable and residuated (particularly Scott-continuous), so for each $x \in X$ fixed, set $O_x:X \to X$ by $O_x(z)=O(z,y)$ for all $y,z \in X$, and for each $y \in X$ fixed, set $O_y:X \to X$ by $O_y(z)=O(x,z)$ for all $x,z \in X$. Hence, by the Theorem \ref{resigalois}, follow the result.\\
	\noindent
	(\ref{(ii)} $\Rightarrow$ \ref{(i)}) It is an immediate consequence of \ref{(i)}.
\end{proof}

A pictorial representation of the Proposition \ref{ou} can be seen in Figure \ref{diagram2}.

\begin{figure}[!htb]
	$$\begin{tikzcd}[row sep=huge, column sep=huge]
	X \times X \arrow{r}[description]{Proj. 2} \arrow{d} [description]{Proj. 1} \arrow[dotted]{dr}[description]{O^{\rho}}
	&  X \arrow[shift right=1]{d}[swap]{O_x ^{\rho}} \arrow[<-,shift left=1]{d}{I_{O_x ^{\rho}}}\\
	X \arrow[<-,shift right=1]{r}[swap]{I_{O_y ^{\rho}}} \arrow[shift left=1]{r}{O_y ^{\rho}}
	& X \\ && X\times X \arrow[dotted]{lu}[description]{I_{O^{\rho}}}
	\arrow[bend left]{llu}
	\arrow[bend right]{uul} 
	\end{tikzcd}$$
	\caption{Diagram of residuated functions}\label{diagram2}
\end{figure}
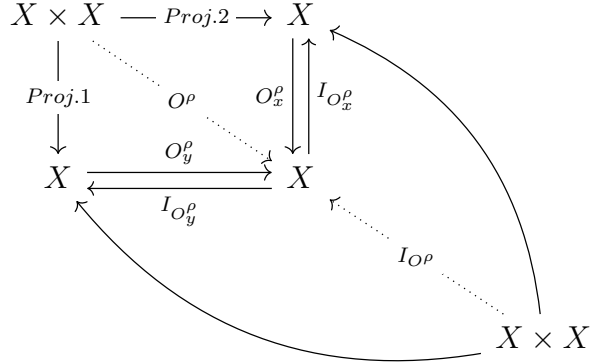

Therefore, given a residuated application $O^{\rho}:X^2 \rightarrow X$  and $z \in X$ the function $\varphi_z :X \rightarrow X$ defined by $\varphi_z (x)= I_{O_x ^{\rho}}(z)$, where $I_{O_x ^{\rho}}$ is the residue of $O_x ^{\rho}:X \rightarrow X$. Similarly define $\psi_z:X \rightarrow X$ by $\psi_z(y)= I_{O_y ^{\rho}}(z)$, where $I_{O_y ^{\rho}}$ is the residue of $O_y ^{\rho}:X \rightarrow X$. 

\begin{corollary}\label{psi}
	Let $X$ a complete lattice of dense order. For all $z\in X$ the following items are worth: 	\begin{enumerate}[labelindent=\parindent, leftmargin=*,label=\textbf{(\roman*)}]
		\item The pair $(\psi_z, \varphi_z)$ forms an adjunction;
		\item The applications $\psi_z$ and $\varphi_z$ are closed operators on $X$.
	\end{enumerate}
\end{corollary}

\section{Final remarks}
In this paper  we propose was the residuation principle for the case of quasi-overlap functions on lattices and their respective induced implications. It has been found that the class of quasi-overlap functions that fulfill the residuation principle is the same class of continuous functions according to Scott's topology. Get a generalization of the residuation principle for quasi-orverlap functions was one of the motivations that led to the writing of this paper, which demonstrated the need for a topology on lattice. This topology is Scott's topology. In fact, given any lattice, it is always possible to know how each element behaves relative to the other elements, but it is difficult to know what the overall structure looks like. However, by defining Scott's topology, the topological properties related to the order that this lattice contains allowed to develop its own visualization for this lattice. Thus, a large number of properties that occur in the closed real interval $ [0,1] $ (eg density, connectivity, as well as being a Hausdorff space) could be generalized to general lattices with specific topological properties. Thus, concepts such as density were expressed both in topological terms and in terms of the defined order relation over the set. Another example was the concept of compactness, which allowed generalizations of the well-known \textit{extreme value theorem} \footnote{The extreme value theorem ensures that a continuous function defined in a compact set reaches its maximum and minimum somewhere in the set.}. The results for any lattice adjunctions using overlap functions allow these operators to be used in tools such as \textit{Mathematical Morphology}, which is applied to the field of signal and image processing through dilation operators, erosion, and others \cite{Haralick}.

In the field of logic, an important point is that residuation is an essential algebraic property that must be required to have good semantics for fuzzy logic systems based on the modus ponens rule, the necessary and sufficient condition for a conjunction fuzzy have a residue is not continuity but left-continuity. Since Scott-continuous quasi-overlap functions are actually a generalization of left-continuous overlap functions to lattices, it is definitely interesting, from a logical point of view, to focus on the study of properties related to left-continuous overlap functions, as well as investigating how these properties are interpreted for Scott-continuous quasi-overlap functions. It is noteworthy that knowledge of left-continuous overlap functions is drastically limited compared to the good description in the continuous case literature.


\end{document}